\documentclass[11pt]{{amsart}}

\usepackage[
  hmarginratio={1:1},     
  vmarginratio={1:1},     
  textwidth=400pt,        
  heightrounded,          
]{geometry}
\usepackage{lipsum}
\usepackage{mathtools}
\usepackage{amsthm}
\usepackage{ dsfont }
\usepackage{amssymb}
\usepackage{enumerate}
\usepackage{graphicx}
\usepackage{float}
\usepackage{bbm}
\usepackage{amsmath}
\usepackage{comment}
\usepackage{hyperref}
\usepackage{listings}
\usepackage{color}
\usepackage{amsmath,amssymb}

\hypersetup{
    colorlinks,
    citecolor=blue,
    filecolor=blue,
    linkcolor=blue,
    urlcolor=blue
    }

\newtheorem{theorem}{Theorem}[section]
\newtheorem{proposition}[theorem]{Proposition}
\newtheorem{lemma}[theorem]{Lemma}
\newtheorem{corollary}[theorem]{Corollary}
\newtheorem{definition}[theorem]{Definition}

\newtheorem{problem}[theorem]{ Problem}
 
\theoremstyle{definition}

\newtheorem*{remark}{Remark}

\def\cA{\mathcal{A}}

\newcommand{\Z}{\mathbb{Z}}
\newcommand{\N}{\mathbb{N}}

\title[]{Universal Spatiotemporal Sampling Sets for Discrete Spatially Invariant Evolutionary Systems}

\author{Sui Tang}
\address{Department of Mathematics}
\curraddr{Johns Hopkins University, Baltimore, MD}
\email{stang.jhu@gmail.com}
\thanks{This research was supported in part by NSF Grant  DMS-1322099.}

\begin{document}

\maketitle
\begin{abstract}
Let $(I,+)$ be a finite abelian group and $\mathbf{A}$ be a circular convolution operator on $\ell^2(I)$.  The problem under consideration is how to construct minimal $\Omega \subset I$ and $l_i$ such that $Y=\{\mathbf{e}_i, \mathbf{A}\mathbf{e}_i, \cdots, \mathbf{A}^{l_i}\mathbf{e}_i: i\in \Omega\}$ is a frame  for $\ell^2(I)$, where $\{\mathbf{e}_i: i\in I\}$ is the canonical basis of $\ell^2(I)$. This problem is motivated by the spatiotemporal sampling problem in discrete spatially invariant evolution systems. We will show that the cardinality of $\Omega $ should be at least equal to the largest geometric multiplicity of eigenvalues of  $\mathbf{A}$, and we consider the universal spatiotemporal sampling sets $(\Omega, l_i)$ for convolution operators $\mathbf{A}$ with eigenvalues  subject to the same largest geometric multiplicity. We will give an algebraic characterization for such sampling sets and show how this problem is linked with sparse signal processing theory and polynomial interpolation theory.
\end{abstract}

\section{Introduction}
 Let $(I,+)$ be a finite abelian group, we denote by $\ell^2(I)$ the space of square summable complex-valued functions defined on $I$ with the inner product given by
 $$\langle \mathbf{f}, \mathbf{g}\rangle =\sum_{i \in I} \mathbf{f}(i)\overline{\mathbf{g}(i)}, \text{ for } \mathbf{f}, \mathbf{g} \in \ell^2(I).$$ We denote by $\{\mathbf{e}_i: i\in I\}$ the canonical basis of $\ell^2(I)$, where $\mathbf{e}_i$ is the characteristic function of $\{i\} \subset I$. 
 
 \begin{definition} The operator $\mathbf{A}$ on $\ell^2(I)$  is called a circular convolution operator if there exists a complex-valued function $\mathbf{a}$ defined on $I$ such that 
$$(\mathbf{A f})(k):=\mathbf{a*f}=\sum_{i\in I}\mathbf{a}(i)\mathbf{f}(k-i), \forall \mathbf{f} \in \ell^2(I) \text { and } k \in I.$$
The function $\mathbf{a}$ is said to be the convolution kernel of $\mathbf{A}$. 
\end{definition}
 
Recall that a finite sequence $\{\mathbf{g}_j\} \subset \ell^2(I)$ is said to be a frame for $\ell^2(I)$, if  there exist positive numbers $A$ and $B$ such that 
 $$A \|\mathbf{f} \|^2 \leq \sum_{j } \lvert \langle \mathbf{f}, \mathbf{g}_j \rangle \rvert^2 \leq B \|\mathbf{f} \|^2, \forall \mathbf{f} \in \ell^2(I).$$

 A finite sequence $\{\mathbf{g}_j\} \subset \ell^2(I)$ is said to be complete if the vector space spanned by $\{\mathbf{g}_j\}$ is $\ell^2(I)$, i.e, 
 if a function  $\mathbf{g} \in \ell^2(I)$ satisfies  $$\langle \mathbf{g}, \mathbf{g}_j \rangle=0, \forall  \mathbf{g}_j,$$
 then $\mathbf{g}=\mathbf{0}.$ In the finite dimensional scenario, it is a standard result of  finite frame theory that a finite sequence $\{\mathbf{g}_j\} \subset \ell^2(I)$ is a frame for $\ell^2(I)$ if and only if it is complete.  
 
 
Let $\mathbf{A}$ be a circular convolution operator on $\ell^2(I)$. We investigate the following problem:
\begin{problem}\label{prob1}
Let $\Omega$ be a proper subset of $I$. We assign a finite nonnegative integer $l_i$ to each $i \in \Omega$. Under what conditions on $\Omega$ and $l_i$ is the sequence 
\begin{equation}\label{seq1}
Y=\{\mathbf{e}_i, \mathbf{A}\mathbf{e}_i,\cdots, \mathbf{A}^{l_i}\mathbf{e}_i: i\in \Omega\}
\end{equation}  a frame for $\ell^2(I)$? In what way does the choice of $\Omega$ and $l_i$ depend on the operator $\mathbf{A}$? 
\end{problem}

Problem \ref{prob1} is motivated by the spatiotemporal sampling and reconstruction problem arising in 
spatially invariant evolution systems \cite{ AADP13, RAS15, AS14, ADK12, ADK13, AKE14, Lv09, JP, JAYM11, JM13}. Let $\mathbf{f} \in \ell^2(I)$ be an unknown vector that is evolving under the iterated actions of a convolution operator $\mathbf{A}$, such that at time instance $t=n$ it evolves to be $\mathbf{A}^n\mathbf{f}$. We call such a discrete evolution system \textit{spatially invariant}. This kind of evolution system may arise as a discretization of a physical process; for example, the diffusion process modeled by the heat equations. We are interested in recovering the initial state $\mathbf{f}$. In practice,  a large number of sensors are distributed to monitor a physical process such as pollution, temperature or pressure \cite{WSN}.  The sensor nodes obtain spatiotemporal samples of physical fields over the region of interest. Increasing the spatial sampling rate is often much more expensive than increasing the temporal sampling rate since the cost of the sensor is more expensive than the cost of activating the sensor \cite{Lv09}.  Given the different costs associated with spatial and temporal sampling, it would be more economically efficient to recover the initial state $\mathbf{f}$ from spatiotemporal samples by using fewer sensors with more frequent acquisition. Following the notation defined in Problem \ref{prob1},  $\Omega$ will be the locations of sensors. At each  $i \in \Omega$, $\mathbf{f}$ is sampled at time instances  $t=0, \cdots, l_i$. Denote by $\mathbf{A}^*$ the adjoint operator of $\mathbf{A}$. By Lemma 1.2 in \cite{ACUS14}, $\mathbf{f}$ can be recovered stably from these spatiotemporal samples if and only if $\overline Y=\{ \mathbf{e}_i, \mathbf{A}^{*}\mathbf{e}_i,\cdots, (\mathbf{A}^{l_i})^{*}\mathbf{e}_i: i\in \Omega \}$ is a frame for $\ell^2(I)$. Note that $\mathbf{A}^*$ is also a circular convolution operator, and its convolution kernel is $\bar{\mathbf{a}}$.  Thus, Problem \ref{prob1} is equivalent to a signal recovery problem in a discrete spatially invariant evolution process.

\subsection{Contribution}
In this paper, we mainly study the cases $I=\Z_d$ and $I=\Z_d \times \Z_d$, where $\Z_d=\{0,1,\cdots,d-1\}$ denotes the finite group of integers modulo $d$ for a  positive integer $d$ and $\Z_d \times \Z_d$ is the product group. In the case of $I=\Z_d$,  we show that the cardinality of $\Omega$ should be at least the largest geometric multiplicity of eigenvalues of $\mathbf{A}$.  We characterize the minimal universal spatiotemporal sampling sets $(\Omega, l_i)$ for  circular convolution operators with eigenvalues subject to the same geometric multiplicity.  Specifically, we show how full spark frames built from the discrete Fourier matrices correspond to minimal universal constructions and demonstrate a close connection with  the sparse signal processing theory.  In the case of $I=\Z_{d} \times \Z_{d}$, we show that finding nontrivial full spark frames from 2D Fourier matrices is less favorable than its 1D sibling. In practice, there are many situations in real applications where the convolution operators have various types of symmetries in frequency response. For these special cases, we employ ideas from interpolation theory of multivariate polynomials to provide specific constructions of minimal sets $\Omega$ and $l_i$ such that Problem \ref{prob1} is solved. The techniques we developed can be adapted to the general finite abelian group case; see related discussions in Section \ref{finiteabeliangroup}.

\subsection{Related Work}
Our work is  closely related to \cite{BAW}, in which the authors have studied the universal spatial sensor locations for discrete bandlimited space; in some sense, finding universal spatiotemporal sampling sets for convolution operators with eigenvalues subject to the same largest geometric multiplicity  in our problem,  is analogous to finding universal spatial sampling sets for discrete bandlimited space $\mathbb{B}^{\mathcal{J}}$ that is subject to the  same cardinality of $\mathcal{J}$ in \cite{BAW}.  However, we do not make sparsity assumptions on the signal space. Instead, we seek sub-Nyquist spatial sampling rate, but want to compensate the insufficient spatial sampling rate by oversampling in time. \\

Our work also has similarities with  \cite{Lv09, JP, JAYM11, JM13}. These works studied the spatiotemporal sampling and reconstruction problem in the continuous diffusion field $\mathbf{f}(x,t)=\mathbf{A}_t\mathbf{f}(x)$, where $\mathbf{f}(x)=\mathbf{f}(x,0)$ is the initial signal and $\mathbf{A}_t$ is the time varying Gaussian convolution kernel determined by the diffusion rule. In \cite{ADK12, ADK13}, Aldroubi and his collaborators develop the mathematical framework of  \textit{Dynamical Sampling} to study the spatiotemporal sampling and reconstruction problem in  discrete spatially invariant evolution processes, which can be viewed as a discrete version of diffusion-like processes. Our results can be viewed as an extension of  \cite{ADK12, ADK13} to the irregular setting and the algebraic characterization given in \cite{ADK12} can be viewed as a special case of our characterization for the unions of periodic constructions.  \\

Other similar works are about  the generalized sampling problems \cite{BAEG13, HAM16, GPV09, GM15, AP77} and the distributed sampling problems \cite{PGH11, CKL08, CYQ15, HRLV10, GGK12}.  For example, in \cite{HAM16}, the authors work in a  $U$-invariant space, and study linear systems $\{L_j:j=1,\cdots,s\}$ such that one can recover any $f$ in  the $U$-invariant space by uniformly downsampling the functions $\{ (L_jf):j=1,\cdots,s\}$, i.e. taking the  generalized samples $\{ (L_jf)(M \alpha)\}_{\alpha \in \Z_d, j=1,\cdots,s}$.


\subsection{Preliminaries and Notation} \label{notation}
In the following, we use standard notation. By $\mathbb{N}$, we denote the set of all positive integers. For $m,n \in \mathbb{N}$, we use $gcd(m,n)$ to denote their greatest common divisor. The linear space of all column vectors with $M$ complex components is denoted by $\mathbb{C}^M$. The linear space of all complex $M \times N$ matrices is denoted by $\mathbb{C}^{M\times N}$.   For a matrix $\mathbf{B} \in \mathbb{C}^{M\times N}$, its transpose is denoted by $\mathbf{B}^T$ and its conjugate-transpose by $\mathbf{B}^*.$ The null space of $\mathbf{B}$ is denoted by $\ker(\mathbf{B})$.  For convenience, both the index of rows and columns for a matrix $\mathbf{B}$ start at 0. If $\mathbf{B} \in \mathbb{C}^{M\times N}$ and $S \subset \{0,1,\cdots,N-1\}$, then let $\mathbf{B}_S$ denote the submatrix of $\mathbf{B}$ obtained by selecting row vectors of $\mathbf{B}$ corresponding to $S$. For example, $\mathbf{B}_{\{0\}}$ means the first row vector of $\mathbf{B}$.  We use $\mathbf{I}_{M\times M}$ to denote the the identity matrix in $\mathbb{C}^{M\times M}$. If $\mathbf{B}$ is a diagonalizable matrix, we denote by $M_{\mathbf{B}}$ the largest geometric multiplicity among all eigenvalues of $\mathbf{B}$ and by $N_{\mathbf{B}}$ the number of distinct eigenvalues of $\mathbf{B}$. We denote by $\mathbf{0}$ the null vector in the vector space. For a vector $\mathbf{z}=(z_i) \in \mathbb{C}^M$, the $M \times M$ diagonal matrix built from $\mathbf{z}$ is denoted by $\mathbf {diag}(\mathbf{z})$. The $\ell^{\infty}$ norm of $\mathbf{z}$ is defined by $\|\mathbf{z}\|_{\ell^{\infty}}=\max_{i}|z_i|$.

For $d \in \N$, we use $\Z_d=\{0,1,\cdots, d-1\}$ to denote the finite group of integers modulo $d$.  If $d=mJ(m>1)$, then $m\Z_d$ denotes the subgroup of $\Z_d$ obtained by selecting the elements divisible by $m$. Let $c$ be an integer between 0 and $m-1$, then $m\Z_d+c=\{c+i: i \in m\Z_d\}$ denotes the  translation of $m\Z_{d}$ by $c$ units, where $``+"$ is  addition modulo $d$. We denote by $S^1$ the unit circle in the plane. With $\omega_d=e^{\frac{-2\pi i}{d}}$ denoting the primitive $d$th root of unity, we define the normalized discrete Fourier matrix via $\mathbf{F}_d=\frac{1}{\sqrt{d}}(\omega_{d}^{jk})_{j,k=0,\cdots d-1}$ and denote  by $ {\mathbf{\hat{f}}}$ the unnormalized discrete Fourier transform of $\mathbf{f}$. For a set $\Omega \subset I$, we use $|\Omega|$ to denote its cardinality.  If $\Omega \subset\{0,\cdots, M-1\}$,  then the subsampling operator $S_{\Omega}$ is  the linear operator that maps  $\mathbf{z} \in  \mathbb{C}^M$ to a vector  in $\mathbb{C}^M$ obtained by modifying the entries of $\mathbf{z}$ whose indices are not in $S$ to be zero.  
 
\begin{definition}\label{levelsets} A set $\Lambda \subset I $ is said to be a level set of $\mathbf{a} \in \ell^2(I)$ if $\mathbf{a}     \equiv
c$ on $\Lambda$ and for all $i \in I- \Lambda$, $\mathbf{a}(i) \neq c$. 
\end{definition}

\begin{definition} Let $\Omega \subset I$. Then the set $\Omega$ is said to be an admissible set for  the convolution operator $\mathbf{A}$ if there exists $l_i \in \N$ for each $i \in \Omega$ such that $Y=\{\mathbf{e}_i, \mathbf{A}\mathbf{e}_i,\cdots, \mathbf{A}^{l_i}\mathbf{e}_i: i\in \Omega\}$ is a frame for $\ell^2(I)$. 
\end{definition}

\begin{definition} Let $\mathbf{v} \in \mathbb{C}^d$ and $\mathbf{B} \in \mathbb{C}^{d\times d}$; then the $\mathbf{B}$-annihilator of $\mathbf{v}$ is the monic polynomial of smallest degree among all the monic polynomials $p$ such that $p(\mathbf{B})\mathbf{v} =\mathbf{0}$. 
\end{definition}
\begin{definition}For $\mathbf{B}=(b_{ij})\in \mathbb{C}^{M \times N}$ and $\mathbf{C}=(c_{ij}) \in \mathbb{C}^{L \times K}$, the Kronecker product of 
$\mathbf{B}$ and $\mathbf{C}$ is defined as the $ML\times NK$ block matrix 
\begin{equation}
\label{Kroneckproduct}
\mathbf{B} \otimes \mathbf{C}=\left[ \begin{array}{ccc} 
b_{0,0}\mathbf{C} & \cdots & b_{0,N-1}\mathbf{C}\\
\vdots & & \vdots \\
b_{M-1,0}\mathbf{C} & \cdots & b_{M-1,N-1}\mathbf{C}\\
\end{array}
\right]. 
\end{equation}
\end{definition}

 The row indices and column indices of $\mathbf{B} \otimes \mathbf{C}$ can be naturally labeled by means of the product group $\Z_M \times \Z_N.$ More precisely, in the row indices $(i_1,i_2)$, $i_1$ is intended to denote on which block row an entry is situated(the row of $\mathbf{B}$), while the index $i_2$ denotes more specifically on which position of its block row it is situated (the row of $\mathbf{C}$). The same rule applies on the column indices. For example, 
the entry of $\mathbf{B} \otimes \mathbf{C}$ at position $((i_1,i_2), (j_1,j_2))$ is $b_{i_1,j_1}c_{i_2,j_2}$.


\subsection{Organization}
The rest of paper is as follows: in Section \ref{prelim}, we present main results for the case $I=\Z_d$.  The algebraic characterization stated in Proposition \ref{prop1}  is the key of solving Problem \ref{prob1}. In Section \ref{twovariablecase}, we consider the case $I=\Z_{d} \times \Z_{d}$. We investigate the similarities and differences with 1D case. A discussion of the general abelian group case is presented in Section \ref{finiteabeliangroup}.  Finally, we summarize our paper in Section \ref{conclusion}.


\section{Single Variable Case $I=\Z_d$}\label{prelim}

In this section, we consider Problem \ref{prob1} for $I=\Z_d$. Motivated by the application in the spatiotemporal sampling problem, we would like to use the least number of sensors to save the budget. So the important issues of Problem \ref{prob1} we want to address  are: 
\begin{enumerate}
\item Given a circular convolution operator $\mathbf{A}$,  what is the minimal cardinality of $\Omega$ such that $Y$ can be a frame for $\ell^2(I)$?
\item Can we find $\Omega$ with minimal cardinality such that $Y$ is a frame?  If so, how to  determine $l_i$  for each $i\in \Omega$? 
\end{enumerate}

In other words, we seek to find minimal admissible sets $\Omega$ for $\mathbf{A}$. Suppose that $\mathbf{A}$ is given by a convolution kernel $\mathbf{a} \in \ell^2(\Z_d)$. We know that  $\mathbf{A}$ admits the spectral decomposition
\begin{equation}\label{equation1}
\mathbf{A} = \mathbf{F}^{\ast}_{d} \mathbf{diag} (\mathbf{\hat a}) \mathbf{F}_{d}. 
\end{equation}
This fact follows from the convolution theorem:
\begin{equation}\label{equation1plus}
\widehat{\mathbf{Af}} = \mathbf{diag (\hat a)} \hat {\mathbf{f}}, \text{ for } \mathbf{f} \in \ell^2(I).
\end{equation}

In  Proposition \ref{prop1}, we will use an algebraic characterization based on the spectral decomposition of $\mathbf{A}$ to answer the questions we just proposed. Proposition \ref{prop1} explains the frame properties for the examples below:  Consider the following two convolution kernels on $\ell^2(\Z_4)$ which define $ \mathbf{A}_1$ and $\mathbf{A}_2$ respectively:

$$\mathbf{\hat a}_1=[1,2, 3, 4]^T, \text{       } \mathbf{ \hat a}_2=[1,2, 1,2]^T.$$

For the convolution operator $\mathbf{A}_1$, Proposition \ref{prop1} shows that for any $i \in \Z_4$, 
$Y = \{\mathbf{e}_i, \mathbf{Ae}_i,\mathbf{A}^2\mathbf{e}_i, \mathbf{A}^{3}\mathbf{e}_i  \}$ is a frame for $\ell^2(\Z_4)$. For the convolution operator $\mathbf{A}_2$, Proposition \ref{prop1}  implies that the cardinality of $\Omega$  should be at least 2.  If we choose $\Omega=\{1,2\}$, then $Y = \{\mathbf{e}_i, \mathbf{Ae}_i:  i\in \Omega \}$ is a frame for $\ell^2(\Z_4)$. However, if we choose $\Omega=\{1,3\}$, then no matter how large $l_i$ is,  $Y = \{\mathbf{e}_i, \mathbf{Ae}_i,\cdots, \mathbf{A}^{l_i}\mathbf{e}_i : i\in \Omega \}$ is never a frame for $\ell^2(\Z_4)$.

Our  characterization can be viewed as a special case of Theorem 2.5 in \cite{ACUS14}, but for convenience we still state it here and prove it.
\begin{proposition}
\label{prop1}
Let $\mathbf{A} = \mathbf{F}^{\ast}_{d} \mathbf{diag} (\mathbf{\hat a}) \mathbf{F}_{d}$ be a circular convolution operator with a kernel $\mathbf{a} \in \ell^2(\Z_d)$, $\{ \lambda_1,\cdots, \lambda_{N_{\mathbf{A}}}\}$ be the set of distinct eigenvalues of $\mathbf{A}$ and  $\{P_k: k=1,\cdots N_\mathbf{A}\}$ be the corresponding eigenspace  projections of $\mathbf{diag} (\mathbf{\hat a})$.  Suppose $\Omega \subset \Z_d$, we define $\mathbf{f}_i =\mathbf{F}_d\mathbf{e}_i$. Let $Y = \{\mathbf{e}_i, \mathbf{Ae}_i,\cdots, \mathbf{A}^{l_i}\mathbf{e}_i : i \in \Omega \}$. 
\begin{enumerate}
\item If $Y$ is a frame for $\ell^2(\Z_d)$, then for each $k$, $\{ P_k(\mathbf{f}_i): i \in \Omega\}$ is a frame for the range space $E_k$ of $P_k$. Hence it is necessary to have $\lvert \Omega \rvert \geq \max_{k=1,\cdots,{N}_{\mathbf{A}}}{dim E_k}$. 
\item For each $i \in \Omega $, we define  by $r_i$ the degree of the $ \mathbf{A}$-annihilator of $\mathbf{e}_i$.  If for each $i \in \Omega$  and each $k$, $l_i \geq r_i-1$ and  $\{ P_k(\mathbf{f}_i): i \in {\Omega} \}$  is a frame for the range space $E_k$ of $P_k$ , then $Y$ is a frame for $\ell^2(\Z_d)$.

\item Let $E$ be the linear span of vectors in $Y$. If for each $i \in \Omega$ and each $k$, $ \mathbf{A}^{l_i+1} \mathbf{e}_i \in E$  and  $\{ P_k(\mathbf{f}_i): i \in {\Omega} \}$  is a frame for the range space $E_k$ of $P_k$, then $Y$ is a frame for $\ell^2(\Z_d)$.
\end{enumerate}
\end{proposition}

\begin{proof} Applying $\mathbf{F}_d$ on $Y$, we obtain a new set of vectors 
$$\tilde{Y}=\{ \mathbf{f}_i, \mathbf{diag} (\mathbf{\hat a})\mathbf{f}_i, \cdots, (\mathbf{diag}(\mathbf{\hat a}))^{l_i}\mathbf{f}_i : i\in \Omega\}.$$
Since $\mathbf{F}_d$ is unitary, $\tilde{Y}$ is a frame for $\ell^2(\Z_d)$ if and only if $Y$ is a frame for $\ell^2(\Z_d)$.

\begin{enumerate}
\item Since we are working with finite dimensional spaces, for each $k$, to show that $\{ P_k(\mathbf{f}_i):  i \in {\Omega} \}$ is a frame for $E_k$, it suffices to show it  is complete. Now suppose that $\mathbf{g} \in E_k$ is orthogonal to vectors in $\{ P_k(\mathbf{f}_i): i \in  {\Omega} \}$, then 
$$\langle \mathbf{g}, \mathbf{diag} (\mathbf{\hat a})^l\mathbf{f}_i \rangle=\langle \mathbf{g}, P_k \mathbf{diag} (\mathbf{\hat a})^l\mathbf{f}_i \rangle=\lambda_j^l \langle \mathbf{g}, P_k(\mathbf{f}_i) \rangle =0$$ for $l=0,\cdots,l_i$ and all $ i \in \Omega$, which means that $\mathbf{g}$ is orthogonal to $\tilde{Y}$. In the above identities, we use the facts that  $P_k\mathbf{g}=\mathbf{g}$, $P_k^2=P_k$ and $P_k \mathbf{diag} (\mathbf{\hat a})=\lambda_k P_k$.  Since $Y$ is a frame, $\tilde{Y}$ is a frame and we conclude that $\mathbf{g}=\mathbf{0}$. Therefore,  $ \{P_k(\mathbf{f}_i): i \in  {\Omega}\}$ is a frame for $E_k$. \\

\item  Given the conditions that  the set of vectors $\{P_k(\mathbf{f}_i): i \in {\Omega} \}$ form a frame for $E_k$ for  $k=1,\cdots, {N}_{\mathbf{A}} $ and $l_i \geq r_i-1$ for $i \in \Omega$,  we will show that $\tilde {Y}$ is complete on $\ell^2(\Z_d)$.  Assume that $\mathbf{g} \in \ell^2(\Z_d)$ is orthogonal to $\tilde{Y}$, i.e., 
$$\langle \mathbf{g}, \mathbf{diag} (\mathbf{\hat a})^l\mathbf{f}_i  \rangle=\mathbf{0}$$ for  $i\in \Omega$ and $l=0,\cdots,l_i$. Since $l_i \geq r_i-1$ and  $r_i$ is the degree of the $\mathbf{A}$-annihilator of $\mathbf{e}_i$ for  $i \in \Omega$, $\mathbf{diag} (\mathbf{\hat a})^l \mathbf{f}_i$ can be written as a linear combination of  vectors in $\tilde{Y}$ for any $l \in \N$ and  we obtain 
$$\langle \mathbf{g},  \mathbf{diag} (\mathbf{\hat a})^l\mathbf{f}_i\rangle=0$$ for  $i\in \Omega$ and $l=0,\dots,N_{\mathbf{A}}-1$.   Note that $\sum\limits_{k=1}^{N_{\mathbf{A}}}{P}_k$ is the identity map, we have the  identity
\begin{equation}
\label {compcd}
\langle  \mathbf{g},  \mathbf{diag} (\mathbf{\hat a})^l\mathbf{f}_i  \rangle=\sum\limits_{k=1}^{N_{\mathbf{A}} }\langle P_j\mathbf{g},  \mathbf{diag} (\mathbf{\hat a})^l\mathbf{f}_i \rangle=\sum\limits_{k=1}^{N_{\mathbf{A}}}\lambda^l_k\langle \mathbf{g}, P_k\mathbf{f}_i\rangle=0
\end{equation}
for  $i\in \Omega$ and $l=0,\dots,N_{\mathbf{A}}-1$. Hence we obtain the following system of linear equations for each $i \in \Omega$:

$$
\left[
  \begin{array}{cccc}
    1& 1&\cdots &1\\
    \lambda_1&\lambda_2& \cdots & \lambda_{N_{\mathbf{A}}}\\
    \vdots &\vdots&  & \vdots \\
    \lambda_1^{N_{\mathbf{A}}-1}&\lambda_2^{N_{\mathbf{A}}-1}& \cdots&\lambda_{N_{\mathbf{A}}}^{N_{\mathbf{A}}-1}  \\
  \end{array}
\right]\left[
  \begin{array}{c}
    \langle \mathbf{g}, P_1\mathbf{f}_i\rangle \\
   \langle \mathbf{g}, P_2\mathbf{f}_i\rangle \\
    \vdots \\
    \langle \mathbf{g}, P_{N_{\mathbf{A}}}\mathbf{f}_i\rangle \\
  \end{array}
\right]=\mathbf{0}. 
$$

 Since the elements of  $\{\lambda_k: k=1,\cdots,N_{\mathbf{A}}\}$ are  distinct, we conclude that $$\langle  \mathbf{g}, P_k\mathbf{f}_i \rangle=\langle P_k \mathbf{g}, P_k\mathbf{f}_i \rangle=0$$ for $k=1, \cdots, N_{\mathbf{A}}$ and $i \in \Omega$. Since $ \{P_k(\mathbf{f}_i): i\in {\Omega}\}$ is a frame for $E_k$, we conclude that $P_k{\mathbf{g}}=\mathbf{0}$ for $k=1,\cdots,N_{\mathbf{A}}$. Therefore,  $\mathbf{g}=\sum\limits_{k=1}^{N_{\mathbf{A}}} P_k{\mathbf{g}}=\mathbf{0}$.  We prove that $\tilde{Y}$ is complete, hence it is a frame and $Y$ is a frame. 

\item The proof is similar to (2). 

\end{enumerate}

\end{proof}


Proposition \ref{prop1} gives a lower bound of the largest geometric multiplicity of eigenvalues of $\mathbf{A}$ for the cardinality of an admissible set $\Omega$. It also tells us in what way the choice of $\Omega$ depends on $\mathbf{A}$ and how to determine $l_i$.  In \cite{ADK12}, the authors have considered the case when $\mathbf{a}$ is a typical low pass filter, i.e,  $\mathbf{\hat {a}}$ is real, symmetric and strictly decreasing on $\{0, 1, \cdots, \frac{d-1}{2}\}$. They prove that $Y$ is not a frame if we choose $\Omega=m\Z_d$ and $l_i=m-1$, and give explicit formulas to construct $\Omega_0 \subset \Z_d-\Omega$ such that $Y \cup \{\mathbf{e}_i: i\in \Omega_0 \}$ is a frame. In this case, by counting the largest geometric multiplicity of eigenvalues of $\mathbf{A}$, we know that the cardinality of an admissible set $\Omega$  should be at least 2.  We use Proposition \ref{prop1} to give a complete characterization of all minimal admissible sets below. 

\begin{corollary}\label{gaussian}
Assume $d$ is odd and $\mathbf{a} \in \ell^2(\Z_d)$ so that $\mathbf{\hat{a}}$ is symmetric and strictly decreasing on $\{0, 1, \cdots, \frac{d-1}{2}\}$. Let  $\mathbf{A}$ be the circular convolution operator with the convolution kernel $\mathbf{a}$.  Suppose that  $\Omega=\{i_1, i_2\} \subset \Z_d$, then 
$$Y=\{\mathbf{e}_i, \mathbf{Ae}_i, \cdots, \mathbf{A}^{\frac{d-1}{2}} \mathbf{e}_i: i\in \Omega \} $$
is a frame for $\ell^2(\Z_d)$ if and only if $gcd(\lvert i_1-i_2 \rvert, d)=1$. 
\end{corollary}
\begin{proof}
First, we  check that the degree of the $\mathbf{A}$-annihilator of $\mathbf{e}_{i_1}$ and $\mathbf{e}_{i_2}$ is $\tfrac{d+1}{2}$. By the symmetry and monotonicity condition of $ {\mathbf{\hat a}}$, we know that the orthogonal eigenspace projections of $\mathbf{diag(\hat{\mathbf{a}})}$ consist of $\{P_j: j=0,\cdots, \tfrac{d-1}{2}\}$,  where $P_j$ is the orthogonal projection onto the subspace spanned by $\{ \mathbf{e}_j, \mathbf{e}_{d-j}\}$ for $j=1,\cdots,\tfrac{d-1}{2}$ and $P_{0}$ is the projection onto the span of $\mathbf{e_0}$. By (2) of Proposition \ref{prop1}, it suffices to show the following $2 \times 2 $ matrix 
\begin{equation}\label{matrix}
\left [\begin{array}{cc} 
\omega_{d}^{i_1j} & \omega_{d} ^{i_2j}\\
\omega_{d}^{i_1(d-j)} & \omega_{d}^{i_2(d-j)}\\
\end{array}
\right]
\end{equation}
is invertible for $j=1,\cdots, \tfrac{d-1}{2}$. We compute their determinants, and know that  they are invertible if and only if $\omega_{d}^{(i_1-i_2)j} \neq 1$ for $j=1,\cdots, \frac{d-1}{2}$. Hence, it is equivalent to the condition $gcd(\lvert i_1-i_2 \rvert, d)=1$. 
\end{proof}
Recall that we denote by  $M_{\mathbf{A}}$ the largest geometric multiplicity of eigenvalues of $\mathbf{A}$. We denote the class of circular convolution operators whose eigenvalues are subject to the same largest geometric multiplicity $L$ 
$$\mathbf{\cA}_{L}=\{ \mathbf{A} \in \mathbb{C}^{d \times d}: \mathbf{Af}=\mathbf{a*f} \text{ for some $\mathbf{a} \in \ell^2(\Z_d)$, } M_{\mathbf{A}} =L \}.$$

By Proposition \ref{prop1}, an admissible set $\Omega \subset \Z_d$ for any  $\mathbf{A} \in \mathcal{A}_L$ must contain at least $L$ elements. Since the spectral projections among $\mathbf{A} \in \mathcal{A}_L$ could be very different, minimal admissible sets  for different $\mathbf{A}$ could be different. But can we find a minimal universal admissible set $\Omega$ with $\lvert \Omega \rvert=L$ for any $\mathbf{A} \in \mathbf{\cA}_{L}$ and then determine $l_i$ for $i \in \Omega$? It turns out  that this question is closely related to full spark frames in the sparse signal processing theory.  We will show this connection but first we need the following definition.  
\begin{definition}
\label{Fullsparkmatrix}
Let $\mathbf{B} \in \mathbb{C}^{M\times N}$. Then the spark of $\mathbf{B}$ is the size of the smallest linearly dependent subset of columns, i.e,
$$Spark(\mathbf{B})=min\{\lvert\lvert \mathbf{f} \rvert\rvert_0: \mathbf{Bf}=\mathbf{0}, \mathbf{f}\neq \mathbf{0}\}.$$
If $M\leq N$, $\mathbf{B}$ is said to be full spark if $Spark(\mathbf{B})=M+1$. Equivalently, $M\times N$ full spark matrices have the property that every $M \times M$ submatrix is invertible. 
\end{definition}

\begin{theorem}\label{fs1} Let $\Omega \subset \Z_d$ with $\lvert \Omega \rvert=L$. Recall that  $N_\mathbf{A}$  is the number of distinct eigenvalues of $\mathbf{A}$. Then, for any $\mathbf{A} \in \mathcal{A}_L$, 
\begin{equation}\label{spatiotemporal}
Y=\{ \mathbf{e_i}, \mathbf{A}\mathbf{e}_i,\cdots, \mathbf{A}^{N_\mathbf{A}-1}\mathbf{e}_i:  i\in \Omega\} 
\end{equation}
is a frame for $\ell^2(\Z_d)$ 
if and only if $(\mathbf{F}_{d})_{\Omega}$ is a $L \times d$ full spark matrix. 
\end{theorem}
\begin{proof} For every $\mathbf{A} \in \mathbf{\cA}_L$, we assume  its convolution kernel is $\mathbf{a}$. Recall Definition \ref{levelsets}, let $\{\Lambda_j: j=1,\cdots,N_\mathbf{A}\}$ be the level sets of $\mathbf{\hat{a}}$. Then we have $\max_{j=1,\cdots,N_\mathbf{A}} \lvert \Lambda_j \rvert=L$. Also let  $\{ P_{\Lambda_j}: j=1,\cdots, N_\mathbf{A}\}$ be the set of eigenspace projections of $\mathbf{diag(\hat{\mathbf{a}})}$, i.e,  $P_{\Lambda_j}$ is the  projection onto  the subspace $E_j$ spanned by $\{\mathbf{e}_l: l\in \Lambda_j \}$.\\

We first prove  the `` if " part.  Let $O_{\Omega}=\{\mathbf{f}_i=\mathbf{F}_d\mathbf{e}_i :  i\in \Omega\}$. We form a $|\Lambda_j| \times L$ matrix with column vectors given by $\{P_{\Lambda_j}\mathbf{f}_i: \mathbf{f}_i\in O_{\Omega} \}$. We claim its  rank is  $|\Lambda_j|$, which implies $\{P_{\Lambda_j}\mathbf{f}_i: \mathbf{f}_i\in O_{\Omega} \}$ is a frame  for $E_j$.  Using the fact that $\mathbf{F}_d=\mathbf{F}^T_d$, it is easy to see that the transpose of the submatrix built from $\{P_{\Lambda_j}\mathbf{f}_i: \mathbf{f}_i\in O_{\Omega} \}$ is a $L \times |\Lambda_j|$ submatrix of $(\mathbf{F}_{d})_{\Omega}$. 
Since $(\mathbf{F}_{d})_{\Omega}$ is a $L \times d$ full spark matrix, any $L \times |\Lambda_j|$ submatrix of $(\mathbf{F}_{d})_{\Omega}$ has rank $|\Lambda_j|$. The claim follows from  this observation. Finally,  we compute the degree of the $\mathbf{A}$-annihilator for $\mathbf{e}_{i}$ which equals to $N_\mathbf{A}$.  By (2) of  Proposition \ref{prop1}, $Y$ defined in \eqref{spatiotemporal} is a frame for $\ell^2(\Z_d)$. \\

Conversely, suppose that $Y$ defined in \eqref{spatiotemporal} is a frame for any $\mathbf{A} \in \mathbf{\cA}_L$. By (1) of Proposition \ref{prop1}, $\{ P_{\Lambda_j} \mathbf{f}_{i}: i\in \Omega \}$ is  a frame for $E_j$. Since the level sets  $\{\Lambda_j: j=1,\cdots, N_{\mathbf{A}}\}$ of a convolution kernel $\mathbf{\hat a}$ for $\mathbf{A} \in \mathbf{\cA}_L$ can have all possibilities of  disjoint partitions of $\Z_d$ satisfying $\max_j \lvert \Lambda_j \rvert =L$, and using the same embedding trick with the ``if" part,  we know that  any $L_1 \leq L$ column vectors of $(\mathbf{F}_{d})_{\Omega}$ must be linearly independent. In other words, $(\mathbf{F}_{d})_{\Omega}$ is a full spark matrix. 
\end{proof}

For example, let us consider the universal minimal constructions for the class of convolution operators $\mathcal{A}_2$. Assume $\Omega=\{i_1,i_2\} \subset \Z_d$. It is direct to check that $(\mathbf{F}_{d})_{\Omega}$ is a $2 \times d $ full spark matrix if and only if $gcd(| i_1-i_2 |,d)=1$. Thus, we get an immediate corollary which generalizes Corollary \ref{gaussian}.

\begin{corollary} Let $\Omega=\{i_1,i_2\} \subset \Z_d$. Then, for any $\mathbf{A} \in \mathcal{A}_2$, 
$$
Y=\{ \mathbf{e_i}, \mathbf{A}\mathbf{e}_i,\cdots, \mathbf{A}^{N_\mathbf{A}-1}\mathbf{e}_i:  i\in \Omega\} 
$$
is a frame for $\ell^2(\Z_d)$ if and only if $gcd(| i_1-i_2 |,d)=1$.
\end{corollary}

Note that $\mathbf{F}_d$ is a Vandemonde matrix, if we choose $\Omega=\{0,1,\cdots L-1\}$, then $(\mathbf{F}_{d})_{\Omega}$ is full spark (\cite[Lemma 2]{ACM12}). Thus, we have the following corollary.

\begin{corollary} Let $ 1\leq L \leq d$ be an integer and $\Omega=\{0,1,\cdots, L-1\}$. Then, for any $\mathbf{A} \in \mathcal{A}_L$, 
$$
Y=\{ \mathbf{e_i}, \mathbf{A}\mathbf{e}_i,\cdots, \mathbf{A}^{N_\mathbf{A}-1}\mathbf{e}_i:  i\in \Omega\} 
$$
is a frame for $\ell^2(\Z_d)$.
\end{corollary}

Full spark matrices play an important role in applications like sparse signal processing, data transmission and phaseless reconstructions. There is a pressing need for deterministic constructions of full spark matrix (\cite{ACM12,  BAW} and the related work in \cite{SM08, SM09}) . In \cite{ACM12}, the authors have considered the problem of finding $\Omega$ such that $(\mathbf{F}_{d})_{\Omega}$ is a full spark matrix.  The following useful properties of full spark matrices can be found in \cite[Theorem 4]{ACM12}. 
\begin{theorem}\label{universalsampling} Let $\Omega \subset \Z_d$. If $(\mathbf{F}_d)_{\Omega}$ is a full spark matrix, then so is the submatrix of $\mathbf{F}_d$ built from the rows indexed by 
\begin{enumerate} 
\item any translations of $\Omega$, $\Omega+r=\{r+i: i \in \Omega\}$. 
\item $r\Omega=\{ri: i \in \Omega\}$ where $r$ is coprime to $d$. 
\item the complement of $\Omega$ in $\Z_d$. 
\end{enumerate}
\end{theorem}

 In general, it is  challenging to give a characterization of finding deterministic full spark matrices from rows of $\mathbf{F}_d$. In \cite{CRT06}, the authors thought that the difficulty may come from the existence of nontrivial subgroups of $\Z_d$. In the 1920s, Chebotar$\ddot{e}$v gave the first characterization to the special case when $d$ is a prime. We can find an introduction and a proof of  this result  in the survey paper  \cite{SLH96}. 
\begin{theorem}[Chebotar$\ddot{e}$v] \label{chebotar}Let d be prime. Then every square submatrix of $\mathbf{F}_d$ is invertible.
\end{theorem}
Later \cite{ACM12} and \cite{BAW} generalized the techniques developed by Chebotar$\ddot{e}$v and gave a characterization to the special case when $d$ is a power of prime. We list the results in \cite{ACM12} here. To understand their results, we need the definition in \cite{ACM12} here. 
\begin{definition} We say a subset $\Omega \subset \Z_d$ is uniformly distributed over the divisors of $d$ if,
for every divisor $m$ of $d$, the $m$ cosets of $\langle m \rangle$ partition $\Omega$ into subsets, each of size 
$\lfloor \frac{\lvert \Omega \rvert}{m} \rfloor$ or $\lceil \frac{\lvert \Omega \rvert}{m} \rceil.$
\end{definition}
For example, when $d$ is prime, every subset of $\Z_d$ is uniformly distributed over the divisors of $d$. 
Then $\{0,1,\cdots,L-1\}$ is uniformly distributed over the divisors of $d$ for any $L\leq d$. The following characterization can be found in \cite[Theorem 9]{ACM12}. 
\begin{theorem} \label{fs}
Let $d$ be a prime power. We select rows indexed by $\Omega \subset \Z_d$ from $\mathbf{F}_d$ to build the submatrix $(\mathbf{F_d})_{\Omega}$. Then $(\mathbf{F_d})_{\Omega}$ is full spark if and only if $\Omega$ is uniformly distributed over the divisors of $d$.
\end{theorem}

As a consequence of Theorem \ref{fs} and Theorem \ref{chebotar}, we can state an immediate corollary.
\begin{corollary} Let $\Omega \subset \Z_d$ with $\lvert \Omega \rvert=L$ and 
$$Y=\{ \mathbf{e_i}, \mathbf{A}\mathbf{e}_i,\cdots, \mathbf{A}^{N_\mathbf{A}-1}\mathbf{e}_i:  i\in \Omega\} .$$
\begin{enumerate}
\item If $d$ is prime, then, for any $\mathbf{A} \in \mathcal{A}_L$, 
$Y$ is a frame for $\ell^2(\Z_d)$. 
\item Assume $d$ is some power of prime. Then,  for any $\mathbf{A} \in \mathcal{A}_L$, $Y$ is a frame for $\ell^2(\Z_d)$ if and only if  $\Omega$ is uniformly distributed over the divisors of $d$.
\end{enumerate}
\end{corollary}

The proof of Theorem \ref{fs} in \cite{ACM12} also shows the necessity of having row indices uniformly
distributed over the divisors of $d$ in general. However, it remains open to prove that it is a sufficient condition for arbitrary $d$. So far, we only know the selections of consecutive rows and their variants given in  Theorem \ref{universalsampling} produce full spark matrices for any  $\mathbf{F}_d$. \\


In the spatiotemporal sampling problem of the heat diffusion process, a common approach is to place the sensors indexed by $\Omega$ in a periodic nonuniform way (\cite{Lv09,JP, JAYM11, JM13} and reference therein). Let $m$ be a  positive divisor of $d$ such that $d=mJ$, and we also investigate when a union of periodic sets $\Omega=\{ m\Z_d+r, r\in W \subset \Z_m\}$ is an admissible set for  an operator $\mathbf{A}$
defined by a convolution kernel $\mathbf{a} \in \ell^2(\Z_d)$. We ask  similar questions in this case: what is the minimal cardinality of $W$ such that $\Omega$ is an admissible set for $\mathbf{A}$ and how to find such $W$. It turns out that the answers are also related to the geometric multiplicity of eigenvalues of $\mathbf{A}$. We define 
$\mathbf{a}_k \in \ell^2(\Z_m)$ by 
\begin{equation}\label{kslice}
\mathbf{a}_k=[{\mathbf{\hat a}}(k), {\mathbf{\hat a}}(k+J),\cdots, {\mathbf{\hat a}}(k+(m-1)J) ]^T,
\end{equation} 
and $\mathbf{D}_k \in \mathbb{C}^{m\times m}$ by
\begin{equation}
\mathbf{D}_k=\mathbf{diag}(\mathbf{a}_k) 
\end{equation} for $k=0,\cdots, J-1.$ Let$$
\mathcal{B}_{L}=\{ \mathbf{A} \in \mathbb{C}^{d \times d}: \mathbf{Af}=\mathbf{a*f} \text{ for some $\mathbf{a} \in \ell^2(\Z_d)$, } \max_{k=0,\cdots,J-1}M_{\mathbf{D}_k} =L \}.$$


We are going to show that for $\mathbf{A} \in \mathcal{B}_L$, the minimal cardinality of $W$ is $L$ and we also find universal admissible unions of periodic sets $\Omega$  with $\lvert W \rvert =L$ for all  $\mathbf{A} \in \mathcal{B}_L$. 

\begin{theorem} \label{main2}Assume $d$ is a positive integer and $m$ is a positive divisor of $d$ such that $d=mJ (m>1)$. Assume that $W\subset \Z_m$ consists of $L<m$ elements. Let $\Omega=\{m\Z_d+r: r\in W\} \subset \Z_d$ and 
\begin{equation}\label{spatiotemporal2}
Y=\{\mathbf{e}_i, \mathbf{A}\mathbf{e}_{i}, \cdots, \mathbf{A}^{m-1}\mathbf{e}_{i}:  i\in \Omega \}.
\end{equation}
\begin{enumerate}
\item If $Y$ is a frame of $\ell^2(\Z_d)$ for  an operator $\mathbf{A} \in \mathcal{B}_L$, then $\lvert W \rvert \geq L$. Moreover, if $\Omega$  is admissible  for  an operator $\mathbf{A} \in  \mathcal{B}_L$, then $\lvert W \rvert \geq L$.
\item $Y$ is a frame for $\ell^2(\Z_d)$ for any $\mathbf{A} \in \mathcal{B}_L$ if and only if the submatix $(\mathbf{F}_{m})_{W}$ is a $\lvert W \rvert \times m$ full spark matrix. 
\end{enumerate}
\end{theorem}\label{periodic}
\begin{proof}
\begin{enumerate}
\item  Assume $\mathbf{A} \in \mathcal{B}_L$ and its convolution kernel is $\mathbf{a}$. Let $\mathbf{f} \in \ell^2(\Z_d)$ and assume that it is orthogonal to $Y$. 

 For a fixed $r \in W$, recall the definition of subsampling operator defined in Subsection \ref{notation}, we let $\mathbf{y}_{s,r} =S_{m\Z_d+r}((\mathbf{A}^{ s})^*\mathbf{f})$ for $s=0,1, \cdots, m-1$.  If we take the discrete Fourier transform on $\mathbf{y}_{s,r}$, and use Poisson Summation Formula and the convolution theorem, then we can obtain identities 
\begin{equation}\label{psf}
{\mathbf{\hat y}}_{s,r}(k)=\frac{1}{m}\sum\limits_{l=0}^{m-1} e^{\frac{2\pi ir(k+Jl)}{d}}  { \overline{\mathbf{ \hat a}(k+Jl)^{s}}} {\mathbf{\hat f}}(k+Jl), k=0,\cdots, J-1
\end{equation} for $s=0,\cdots,m-1$. For  $k=0,1,\cdots, J-1$, we define  
$$
\mathbf{A}_{m,k} =\left [\begin{smallmatrix}
1 & 1 &\cdots & 1\\
\overline{\mathbf{\hat a} (k)} & \overline{\mathbf{\hat a}(k+J)}&\cdots & \overline{\mathbf{\hat a} (k+(m-1)J)}   \\
\vdots &\vdots & &\vdots \\
\overline{\mathbf{\hat a}^{m-1}(k)} & \overline{\mathbf{\hat a}^{m-1}(k+J)}&\cdots & \overline{\mathbf{\hat a}^{m-1}(k+(m-1)J)}
\end{smallmatrix}
\right] \text{ and } \mathbf{h}_{r,k}=\left [\begin{smallmatrix}
e^{\frac{2\pi irk}{d}} \\
e^{\frac{2\pi ir(k+J)}{d}}  \\
\vdots  \\
e^{\frac{2\pi ir(k+(m-1)J)}{d}} \end{smallmatrix}
\right].
$$ We also define 
$$
\mathbf{y}_{k}^{(r)}=\left [\begin{smallmatrix}
\mathbf{ \hat y}_{0,r}(k) \\
\mathbf{\hat y}_{1,r}(k) \\
\vdots  \\
\mathbf{\hat y}_{m-1,r}(k) \end{smallmatrix}
\right] \text { and } \mathbf{f}_{k}=\left [\begin{smallmatrix}
\mathbf{\hat f}(k) \\
\mathbf{\hat f}(k+J) \\
\vdots  \\
\mathbf{\hat f}(k+(m-1)J) \end{smallmatrix}
\right]. 
$$ For a fixed $k$ and $r$, we put identities  \eqref{psf} with $s=0,1,\cdots,m-1$ into a matrix equation  and obtain 
\begin{equation}
m\mathbf{y}_{k}^{(r)}={\mathbf{A}_{m,k}}\mathbf{diag}(\mathbf{h}_{r,k})\mathbf{f}_{k}.  
\end{equation}By the assumption that  $\mathbf{f}$ is orthogonal to ${Y}$, for each $s$, $\mathbf{y}_{s,r}=\mathbf{0}$ and hence $ \mathbf{\hat y}_{s,r}=\mathbf{0}$. So $\mathbf{y}_k^{(r)}=\mathbf{0}$, we conclude that $\mathbf{f}_{k} \in \ker(\mathbf{A}_{m,k}\mathbf{diag}(\mathbf{h}_{r,k}))$.  Let $r$ take values in all elements of $W$ and $k=0,1,\cdots,J-1$, we see that ${\mathbf{f}}=\mathbf{0}$ if and only if 
\begin{equation}\label{equation4}
\mathop{\bigcap}\limits_{r \in W} \ker(\mathbf{A}_{m,k}\mathbf{diag}(\mathbf{h}_{r,k}))=\{\mathbf{0}\}\end{equation}
for $k=0,\cdots J-1$. Let $\ker(\mathbf{A}_{m,k})^{\bot}$ denote the orthogonal complement of $\ker(\mathbf{A}_{m,k})$ in $\ell^2(\Z_m)$, using the basic linear algebra,  we know that showing \eqref{equation4}  is equivalent to showing 
\begin{equation}\label{equation5}
\sum \limits_{r \in W} (\ker(\mathbf{A}_{m,k}\mathbf{diag}(\mathbf{h}_{r,k})))^{\bot}=\ell^2(\Z_m).
\end{equation}
for  $k=0,\cdots J-1$. 

Recall the definition of $\mathbf{a}_k$ for  the kernel $\mathbf{a}$ in \eqref{kslice}, we let  $\{\Lambda_{j,k}: j=1,\cdots, n_k\}$ be the level sets of $\mathbf{a}_k$. Note that  its complex conjugate  $\overline{\mathbf{a}_k}$ has the same level sets with $\mathbf{a}_k$.  Let $\{P_{\Lambda_{j,k}}: j=1,\cdots, n_k \}$ be the orthogonal projections determined by $\{\Lambda_{j,k}: j=1,\cdots, n_k\}$, i.e,  $P_{\Lambda_{j,k}}$ is the orthogonal projection onto the subspace of $\ell^2(\Z_m)$ spanned by $\{\mathbf{e}_i \in \ell^2(\Z_m): i\in \Lambda_{j,k}\}$.  Let  $\mathbf{v}=[1,1,\cdots,1]^T \in \ell^2(\Z_m)$, observing that $\mathbf{A}_{m,k}$ is a Vandermonde matrix, it is not difficult to see that $\{P_{\Lambda_{j,k}}\mathbf{v}: j=1,\cdots,n_k\}$ is an orthogonal basis for $\ker(\mathbf{A}_{m,k})^{\bot}$. Next, using the relation 
$$\ker(\mathbf{A}_{m,k}\mathbf{diag}(\mathbf{h}_{r,k}))^{\bot}=\mathbf{diag}^*(\mathbf{h}_{r,k})\ker(\mathbf{A}_{m,k})^{\bot},$$
we let $\mathbf{b}_r=\mathbf{F}_m\mathbf{e}_r$ for $r \in W$ and then  we can see that $\{P_{\Lambda_{j,k}}\mathbf{b}_r: j=1,\cdots,n_k\}$ is an orthogonal basis for $\ker(\mathbf{A}_{m,k}\mathbf{diag}(\mathbf{h}_{r,k}))^{\bot}$. Hence, for each $k$, showing \eqref{equation5} is equivalent to showing 
\begin{equation}\label{equation6}
\{P_{\Lambda_{j,k}}\mathbf{b}_r: r \in W, j=1,\cdots,n_k\}
\end{equation} is complete on $\ell^2(\Z_m)$.  Note that $\{P_{\Lambda_{j,k}}: j=1,\cdots, n_k \}$ is a set of pairwise orthogonal projections, \eqref{equation6} is true if and only if 
$$\{P_{\Lambda_{j,k}}\mathbf{b}_r: r \in W\}$$ is complete on the range space $E_{j,k}$ of $P_{\Lambda_{j,k}}$ for $j=1,\cdots, n_k$.  Now let us summarize what we have proved: 
assume that $ \mathbf{f} \in \ell^2(\Z_d)$ and it is orthogonal to $Y$, then $\mathbf{f}=\mathbf{0}$ if and only if for each $k=0,1,\cdots, J-1$,
\begin{equation}\label{equation7}
\{P_{\Lambda_{j,k}}\mathbf{b}_r: r \in W \}
\end{equation}
is complete on the range space  $E_{j,k}$ of $P_{\Lambda_{j,k}}$ for $j=1,\cdots, n_k$.  By the definition of $\mathcal{B}_L$,  $\max_{j,k} dim(E_{j,k})=L$. Given the condition that $Y$ is a frame for $\ell^2(\Z_d)$, we conclude that  $\rvert W\lvert \geq L$. \\

We point out that the condition $l_i=m-1$ is not  essential here; indeed, we prove if $\Omega=\{m\Z_d+r, r\in W\}$ is an admissible set for an operator  $\mathbf{A} \in \mathcal{B}_L$, then $\rvert W\lvert \geq L$.  Since if $Y$ defined in $\eqref{spatiotemporal2}$ is not a frame for $\ell^2(\Z_d)$, then no matter how large we increase each $l_i$ for $i \in \Omega$, the  new obtained $Y$  will never be a frame. This fact follows from the special structure of Vandermonde matrix.

\item  Using the characterization summarized in \eqref{equation7} and the same argument as in the proof of Theorem \ref{fs1}, Y defined in $\eqref{spatiotemporal2}$ is a frame for any $\mathbf{A} \in \mathcal{B}_L$ if and only if $(\mathbf{F}_{m})_{W} $ is full spark.

\end{enumerate}

 \end{proof}

\begin{remark} Proposition 3.1 in \cite{ADK12} says that $W=\{0\}$ will be an admissible set for all $\mathbf{A} \in \mathcal{B}_1$, which can be viewed as  this theorem's special case.  In fact, this theorem shows that any
$W \subset \Z_m$ with $\lvert W\rvert=1$ is an admissible set for $\mathcal{B}_1$.  
\end{remark}

 As an immediate corollary, we get 

\begin{corollary} Suppose we have the same settings with Theorem \ref{main2}. 
\begin{enumerate}
\item If $W=\{0,1,\cdots, L-1\}$, then for any $\mathbf{A} \in \mathcal{B}_{L}$, $Y$ is a frame for $\ell^2(\Z_d)$. 
\item If m is prime, then for any $\mathbf{A} \in \mathcal{B}_{L}$, $Y$ is a frame for $\ell^2(\Z_d)$. 
\item If $m$ is a  power of prime and $W$ is uniformly distributed over the divisor of $m$,  then  for any  $\mathbf{A} \in \mathcal{B}_{L}$,  $Y$ is a frame for $\ell^2(\Z_d)$.
\end{enumerate}
\end{corollary}

Since $m$ is a divisor of $d$, we can choose $m$ to be prime or some power of a prime. If this is the case, we immediately know how to construct all possible $W$ to give an admissible union of periodic set $\Omega$ for $\mathcal{B}_L$.



\section{Two Variable Case $I=\Z_{d}\times \Z_{d}$} \label{twovariablecase}
In this section, we consider the case $I=\Z_d\times \Z_d$, which is the product group of two identical groups $\Z_d$. Suppose $\mathbf{A}$ is a circular convolution operator defined by a convolution kernel $\mathbf{a} \in \ell^2(\Z_d \times \Z_d)$.  One natural way to consider Problem \ref{prob1} in the two variable setting  is viewing it as a single variable setting since we can map $\Z_d \times \Z_d$  to $\Z_{d^2}$ by sending $(i,j)$ to $di+j$. Under this identification, the operator $\mathbf{A}$ corresponds to a linear operator $\tilde{{\mathbf {A}}}$ acting on $\ell^2(\Z_{d^2})$ and $\tilde{\mathbf {A}}=(\mathbf{F}_d \otimes \mathbf{F}_d)^{*} \mathbf{diag(\tilde{\hat{a}})} (\mathbf{F}_d \otimes \mathbf{F}_d)$, where $\mathbf{\tilde{\hat{a}}}$ is the image of ${ \mathbf{ \hat a}}$ in $\ell^2(\Z_{d^2})$ under the above identification. Using the labelling method described in Subsection \ref{notation}, we state a similar version of Proposition $\ref{prop1}$ for the two variable case.

\begin{proposition}
\label{prop2}
Let $\mathbf{A}$ be a circular convolution operator with a kernel $\mathbf{a} \in \ell^2(\Z_d\times \Z_d)$ and $\{\Lambda_k: k=1,\cdots, N_{\mathbf{A}}\}$ be the level sets of ${\mathbf{\hat a}}$. Suppose  $\Omega \subset \Z_d \times \Z_d$. Let 
 $Y = \{\mathbf{e}_{j_1,j_2}, \mathbf{Ae}_{j_1,j_2},\cdots, \mathbf{A}^{l_{j_1,j_2}}\mathbf{e}_{j_1, j_2} : (j_1,j_2) \in \Omega \}$. 
\begin{enumerate}
\item If $Y$ is a frame for $\ell^2(\Z_d\times \Z_d)$, then for each $k$, the submatrix of $\mathbf{F}_d\times \mathbf{F}_d$ built from rows indexed by $\Lambda_k$ and columns indexed by $\Omega$  has rank $|\Lambda_k|$.  Hence it is necessary to have $\lvert \Omega \rvert \geq \max_{k}|\Lambda_k|$. 
\item For each $(j_1,j_2) \in \Omega $, we define  by $r_{j_1,j_2}$  the degree of the $ \mathbf{A}$-annihilator of $\mathbf{e}_{j_1,j_2}$. If for each $(j_1,j_2) \in \Omega$ and each $k$, $l_{j_1,j_2} \geq r_{j_1,j_2}-1$  and the submatrix of $\mathbf{F}_d\times \mathbf{F}_d$ built from rows indexed by $\Lambda_k$ and columns indexed by $\Omega$ has rank $|\Lambda_k|$, then  $Y$ is a frame for $\ell^2(\Z_d\times \Z_d)$.
\item If for each $(j_1, j_2) \in \Omega$ and each $k$, $\mathbf{A}^{l_{j_1,j_2}+1} \mathbf{e}_{j1,j2}$ is in the space spanned by Y and  the submatrix of $\mathbf{F}_d\times \mathbf{F}_d$ built from rows indexed by $\Lambda_k$ and columns indexed by $\Omega$  has rank  $|\Lambda_k|$,  then  $Y$ is a frame for $\ell^2(\Z_d\times \Z_d)$. 
\end{enumerate}
\end{proposition}
Similar to the one variable case, the problem of finding a common minimal admissible set $\Omega$ for  2D convolution operators  with eigenvalues  subject to the same largest geometric multiplicity is equivalent to finding full spark matrices from rows of $\mathbf{F}_d \otimes \mathbf{F}_d$ indexed by $\Omega$. Unlike $\mathbf{F}_d$, $\mathbf{F}_d \otimes \mathbf{F}_d$ is not a Vandermonde matrix. In the case of $\mathbf{F}_d$,  the submatrix built from any consecutive $L \leq d$ rows of $\mathbf{F}_d$ is full spark. This fact follows from the Vandermonde structure of $\mathbf{F}_d$. But the following lemma shows that this is not true for $\mathbf{F}_d \otimes \mathbf{F}_d$ . In fact, we prove that it requires at least $d+1$ rows of $\mathbf{F}_d \otimes \mathbf{F}_d$ to build a nontrivial full spark matrix. 
\begin{lemma}
\label{lemt}
For a positive integer $1<L\leq d$, given any $L$ rows of $\mathbf{F}_d \otimes \mathbf{F}_d$, there exist $L$ columns of  $\mathbf{F}_d \otimes \mathbf{F}_d$ such that the resulting $L \times L$ submatrix is not invertible. 
\end{lemma}
\begin{proof}Let $\{(k_j,l_j):j=1,\cdots L\}$ be the row indices  of $\mathbf{F}_d \otimes \mathbf{F}_d$. If $ 1 <L\leq d$, we claim that there exist column indices  $\{(s_j,p_j):j=1,\cdots,L\}$ such that the resulting $L \times L$ submatrix has two identical rows. Let $G$ be the additive subgroup of $Z_d\times \Z_d$ generated by $(k_1-k_2,l_1-l_2)\in \Z_d \times \Z_d$, then $\lvert G \rvert \leq d.$ It is known from  Pontryagin duality theory (see \cite{d}) that there exists a corresponding annihilator subgroup $H \subset \Z_d \times \Z_d$ of size $\frac{d^2}{\lvert G \rvert}$, such that for any $(s,p) \in H$,
\begin{equation}\label{duality}
\omega_{d}^{(k_1-k_2)s+(l_1-l_2)p}=1.
\end{equation}
Since  $\lvert H\rvert=\frac{d^2}{\lvert G \rvert} \geq d$, we can choose any subset of $H$ consisting of $L$ elements as column indices.  By \eqref{duality},
$$\omega_{d}^{k_1s+l_1p}=\omega_{d}^{k_2s+l_2p}, $$ 
which means that the first two rows of the built submatrix are identical.  
\end{proof}

We can also prove that for some $2D$ convolution operators, it may happen that the lower bound for the cardinality of  admissible sets is more than the largest geometric multiplicity of their eigenvalues. While we know a construction of full spark matrix built from rows of $\mathbf{F}_d$ with any given spark between $1$ and $d$, Lemma \ref{lemt} shows that it is impossible to find a full spark matrix with spark less than $(d+1)$ from submatrices built from rows of $\mathbf{F}_d \otimes \mathbf{F}_d.$ To our best knowledge, in the case of $\mathbf{F}_d \otimes \mathbf{F}_d$, there is no deterministic formula that gives a way to construct full spark matrices for any $d$. However, we can draw a similar conclusion as Theorem \ref{universalsampling}. 
\begin{proposition}If the submatrix  of $ \mathbf{F}_d \otimes \mathbf{F}_d$ built from rows  indexed by $\Omega \subset \Z_d \times \Z_d$ is full spark, so is the submatrix of $\mathbf{F}_d \otimes \mathbf{F}_d$  build from rows indexed by 
\begin{enumerate}
\item $\Omega+(s,p)$ for any $(s,p) \in \Z_d\times \Z_d$,
\item $(c_1,c_2)\Omega=\{(c_1i, c_2j): (i,j) \in \Omega\}$ for any $(c_1, c_2) \in \Z_d\times \Z_d$ such that both $c_1, c_2$ are coprime to $d$, 
\item $\Omega^c=\Z_d\times \Z_d-\Omega$. 
\end{enumerate}
\end{proposition}

In modeling physical or biological phenomena, the convolution kernel $\mathbf{a}$ usually  possesses certain symmetries in the frequency domain. We introduce several types of symmetric convolution kernels and consider the problem of finding minimal admissible $\Omega$ for these symmetric convolution kernels. For the convenience of statement, we assume $d$ is odd and we identify the level sets of $\mathbf{\hat a}$ with their modulo images in $\{-\frac{d-1}{2}, \cdots, \frac{d-1}{2}\} \times \{-\tfrac{d-1}{2}, \cdots, \tfrac{d-1}{2}\}$.  In the rest of paper, we denote $\mathcal{I}=\{-\frac{d-1}{2}, \cdots, \frac{d-1}{2}\} \times \{-\tfrac{d-1}{2}, \cdots, \tfrac{d-1}{2}\}$. The following definitions can be found in \cite{HIP}. 
\begin{definition}
\label{def4}Let $\mathbf{a}$ be a $2D$ array defined on $\Z_{d}\times \Z_{d}$. Recall that  ${\mathbf{\hat a}}$ denote its unnormalized discrete Fourier transform. 

\begin{enumerate}
\item Given any  level set $\Lambda$ of ${\mathbf{\hat a}}$, if all elements in $\Lambda$ have the  same $\ell^{\infty}$ norm, $\mathbf{a}$ is said to possess $\ell^{\infty}$-symmtery in frequency response.
\item  If the level sets of ${\mathbf{\hat a}}$ consist of the sets in the form  of $\{(s,p),(s,-p),(-s,p),(-s,-p)\}$ for $(s,p) \in \mathcal{I}$, $\mathbf{a}$ is said to possess quadrantal symmetry in frequency response.
\item If the level sets of ${\mathbf{\hat a}}$ consist of the sets in the form of  $\{(p,s),(s,p),(-p,-s),(-s,-p)\}$ for $(s,p) \in \mathcal{I}$, $\mathbf{a}$ is said to possess diagonal symmetry in frequency response.
\item If the level sets of ${\mathbf{\hat a}}$  consist of the sets in the form of  $$\{ (s,p),(p,s),(-p,s),(-s,p), (-s,-p), (-p,-s),(p,-s),(s,-p)\}$$
for $(s,p) \in \mathcal{I}$, $\mathbf{a}$ is said to possess octagonal symmetry in frequency response.
\end{enumerate}
\end{definition}

For convolution kernels $\mathbf{a}$ with the same symmetry in  frequency response,  their discrete Fourier transform have the same level sets. By Proposition \ref{prop2}, it is possible for us to construct minimal admissible  sets  and then determine $l_{i,j}$ for 2D convolution kernels that are subject to the same symmetry conditions in frequency response. The specific constructions we will show are inspired by the ideas stemmed from the multivariable interpolation theory. 
\begin{theorem}\label{stablerecovery4case}
Let $\mathbf{A}$ be a circular convolution operator defined by a convolution kernel $\mathbf{a} \in \ell^2(\Z_d\times \Z_d)$. Let $\Omega$ be a proper subset of $\Z_d\times \Z_d$ and $Y = \{\mathbf{e}_{i,j}, \mathbf{Ae}_{i,j},\cdots, \mathbf{A}^{l_{i,j}}\mathbf{e}_{i,j} : (i,j) \in \Omega \}$. 
\begin{enumerate}
\item Assume that $\mathbf{a}$ possesses $\ell^{\infty}$ symmetry in frequency response. If  we choose $\Omega=\{ 0,1\} \times \Z_d \cup \Z_d \times \{0,1\}$ and for each $(i,j) \in \Omega$, $l_{i,j}=\frac{d-1}{2}$, then $Y$ is a frame for $\ell^2(\Z_d\times \Z_d)$. 
\item Assume that $\mathbf{a}$ possesses quadrantal symmetry in frequency response. Suppose  the elements $i_1, i_2, j_1$ and $ j_2 $ of  $\Z_d$ satisfy the conditions $gcd(\lvert i_1-i_2 \lvert ,d)=1$ and $gcd(\lvert j_1-j_2 \rvert, d)=1$.  If we  choose $\Omega=\{(i_1,j_1),(i_2,j_2),(i_1,j_2),(i_2,j_1)\}$, and for each $(i,j)\in \Omega$, $l_{i,j}=\frac{(d+1)^2}{4}-1$, then $Y$ is a frame  for $\ell^2(\Z_d \times \Z_d)$.
\item Assume that  $\mathbf{a}$ possesses diagonal symmetry in frequency response. Suppose  the elements $i_1, i_2 , j_1$ and $j_2$ of $\Z_d$  satisfy the conditions $gcd(i_2,d)=1$ and $gcd(j_2,d)=1$. If we choose  $\Omega=\{(i_1,0),( i_1+i_2,0),(i_1+2i_2,0),(i_1+3i_2,0)\} \text{ or } \Omega=\{(0, j_1),(0,j_1+j_2),(0,j_1+2j_2), (0,j_1+3j_2)\}$, and for each $(i,j) \in \Omega$, $l_{i,j}=\frac{(d+1)^2}{4}-1$, then $Y$ is a frame  for $\ell^2(\Z_d \times \Z_d)$.
\item Assume that  $\mathbf{a}$ possesses octagonal symmetry in frequency response. Suppose the elements $i_1, i_2 , j_1$ and  $j_2 $ of  $\Z_d$ satisfy the conditions $gcd(\lvert i_1-i_2\vert,d)=1$ and $gcd(j_2,d)=1$. If we choose $\Omega=\{i_1,i_2 \} \times \{j_1, j_1+j_2,j_1+2j_2, j_1+3j_2 \}$, and for each $(i,j) \in \Omega$, $l_{i,j}=\frac{(d+1)(d+3)}{8}-1$, then $Y$ is a frame  for $\ell^2(\Z_d \times \Z_d)$. Alternatively, 
suppose the elements $i_1, i_2 , j_1$ and $ j_2$ of $\Z_d$ satisfy the conditions $gcd(\lvert j_1-j_2\vert,d)=1$ and $gcd(i_2,d)=1$. If we choose
$ \Omega=\{i_1, i_1+i_2,i_1+2i_2, i_1+3i_2 \} \times \{j_1,j_2\}$, and for each $(i,j) \in \Omega$, $l_{i,j}=\frac{(d+1)(d+3)}{8}-1$, then $Y$ is a frame  for $\ell^2(\Z_d \times \Z_d)$. 

\end{enumerate}
\end{theorem}
\begin{proof}
\begin{enumerate}
\item Since $\mathbf{a}$ has $\ell^{\infty}$ symmetry in frequency response,  ${\mathbf{\hat a}}$ has $\frac{d+1}{2}$ level sets which are given by $\Lambda_l=\{(s,p): (s,p)\in \mathcal{I}, \max(\lvert s\rvert, \lvert p\rvert)=l\}$ for $l=0,\cdots,\frac{d-1}{2}.$ 
Then by computation,  we get $N_{\mathbf{A}}=\frac{d+1}{2}$ and $M_{\mathbf{A}}=4d-4.$  Let $\Omega=\{ 0,1\} \times \Z_d \cup \Z_d \times \{0,1\}$. We first show that the $(4d-4) \times (4d-4)$ submatrix $\mathbf{S}$ built from  rows indexed by $\Lambda_{\tfrac{d-1}{2}}$ and columns indexed by $\Omega$ of $\mathbf{F_d}\otimes \mathbf{F}_d$ is invertible. Assume that there exists a vector $\textbf{c}=(\mathbf{c}(k,l))_{(k,l)\in \Omega}$ such that $\mathbf{S}\textbf{c}=\mathbf{0}.$ For each $(s,p)\in \Lambda_{\tfrac{d-1}{2}}$, we have
\begin{equation}\label{sum0}
\sum\limits_{(k,l)\in \Omega} c(k,l)\omega_d^{ks}\omega_d^{lp}=0.
\end{equation}Reordering the terms in \eqref{sum0} by collecting the coefficients together for the power of $\omega_d^s$, we get 
\begin{equation}
\label{rowequation1}
\sum\limits_{l=0}^{d-1}c(0,l)\omega_d^{lp}+(\sum\limits_{l=0}^{d-1}c(1,l)\omega_d^{lp})\omega_d^{s}+\sum\limits_{k=2}^{d-1}(c(k,0)+c(k,1)\omega_d^{p})(\omega_d^{s})^k=0.
\end{equation}Denote $\Lambda'_{\tfrac{d-1}{2}}= \{-\frac{d-1}{2},\cdots,\frac{d-1}{2}\} \times \{-\frac{d-1}{2},\frac{d-1}{2}\}$. Then obviously  $\Lambda'_{\tfrac{d-1}{2}}\subset \Lambda_{\tfrac{d-1}{2}}$.
Since $\eqref{rowequation1}$ holds for $(s,p) \in \Lambda'_{\tfrac{d-1}{2}}$, we get 
\begin{equation}
\label{transformequations}
\left[ \begin{array}{cccc} 
1 &\omega_d^{-\frac{d-1}{2}} & \cdots &( \omega_d^{-\frac{d-1}{2}})^{d-1}\\
1& \omega_d^{-\frac{d-3}{2}}& \cdots &(\omega_d^{-\frac{d-3}{2}})^{d-1}\\
\vdots & \vdots &\ddots &\vdots \\
1&\omega_d^{\frac{d-1}{2}}&\cdots &(\omega_d^{\frac{d-1}{2}})^{d-1} \\
\end{array}
\right]\left[\begin{smallmatrix}
\sum\limits_{l=0}^{d-1}\mathbf{c}(0,l)\omega_d^{lp} \\
\sum\limits_{l=0}^{d-1}\mathbf{c}(1,l)\omega_d^{lp}\\
\vdots\\
\mathbf{c}(d-1,0)+\mathbf{c}(d-1,1)\omega_d^{p}
\end{smallmatrix} \right]=\mathbf{0}
\end{equation}
for $p=-\frac{d-1}{2}$ and $\frac{d-1}{2}$. 
Note that the matrix on the left of \eqref{transformequations} is an invertible Vandermonde matrix, therefore the coefficient vector on the right of \eqref{transformequations} is $\mathbf{0}$. Then we obtain  
\begin{equation}\label{finalequation}
\sum\limits_{l=0}^{d-1}\mathbf{c}(i,l)(\omega_d^{l})^{p}=0 \text{ for } i=0,1\text{ and } p=\pm \tfrac{d-1}{2}, 
\end{equation}
and
$$
\left[ \begin{smallmatrix}
1 &\omega_d^{-\tfrac{d-1}{2}} \\
1& \omega_d^{\tfrac{d-1}{2}}\\
\end{smallmatrix} \right] \left [\begin{array}{c} 
\mathbf{c}(k,0) \\
\mathbf{c}(k,1)\\
\end{array} \right]=\mathbf{0}
$$ for $ 2 \leq k \leq d-1$.  Solving  the above linear equations, we  get  $\mathbf{c}(k,0)=\mathbf{c}(k,1)=0$ for $2\leq k \leq d-1.$ We can also reorder the terms in \eqref{sum0} by collecting the coefficients together for the power of $\omega_d^p$ and get 
\begin{equation}
\label{rowequation2}
\sum\limits_{k=0}^{d-1}\mathbf{c}(k,0)\omega_d^{sk}+(\sum\limits_{k=0}^{d-1}\mathbf{c}(k,1)\omega_d^{sk})\omega_d^{p}+\sum\limits_{l=2}^{d-1}(\mathbf{c}(0,l)+\mathbf{c}(1,l)\omega_d^{s})(\omega_d^{p})^l=0.
\end{equation}
Since equation $\eqref{rowequation2}$ holds for $(s,p) \in \Lambda'_{\tfrac{d-1}{2}}$, similar to the case of reordering with the power of $\omega_d^s$, we get $\mathbf{c}(0,l)=\mathbf{c}(1,l)=0$ for $2\leq l\leq d-1.$ Now  substituting the entries of $\mathbf{c}$ we just solved into \eqref{finalequation}, and we can get $\mathbf{c}(0,0)=\mathbf{c}(0,1)=\mathbf{c}(1,0)=\mathbf{c}(1,1)=0.$ Thus $\textbf{c}=\mathbf{0}$, which implies that $\mathbf{S}$ is invertible. The same arguments can show that the $8l \times 8l$ submatrix of $\mathbf{F}_d \otimes \mathbf{F}_d$ built by rows indexed by $\Lambda_l$ and $\{0,1\} \times \{0,\cdots,2l\} \cup \{0,\cdots,2l\} \times \{0,1\} \subset \Omega$ is invertible for $l=1,\cdots,\frac{d-3}{2}.$ Therefore, we have proved that the submatrix of $\mathbf{F}_d \otimes \mathbf{F}_d$ formed by rows indexed by $\Lambda_l$ and columns indexed by $\Omega$ is of full row rank for $l=0,\cdots,N_\mathbf{A}-1.$ By (2) of Proposition \ref{prop2}, the conclusion follows. 
\item Assume that  $\mathbf{a}$ has quadrantal symmetry in frequency response, it can be easily computed that $M_{\mathbf{A}}=4$ and $N_{\mathbf{A}}=\frac{(d+1)^2}{4}$. Let $$\Lambda_{s,p}=\{(s,p), (s,-p), (-s,-p), (-s,p)\}$$ be a level set of $\mathbf{\hat a }$. Given $\Omega=\{\{i_1,i_2\} \times\{j_1,j_2\}: gcd(\lvert i_1-i_2 \lvert ,d)=1, gcd(\lvert j_1-j_2 \rvert, d)=1\}$, the submatrix built from rows indexed by $\Lambda_{s,p}$ and columns indexed by $\Omega$ for the case $s\neq 0 \text{ and } p\neq 0$ is 
$$
\left[ \begin{smallmatrix}
\omega_d^{si_1+pj_1} &\omega_d^{si_1+pj_2}&\omega_d^{si_2+pj_2}&\omega_d^{si_2+pj_2} \\
\omega_d^{si_1-pj_1}& \omega_d^{si_1-pj_2}&\omega_d^{si_2-pj_2}&\omega_d^{si_2-pj_2}\\
\omega_d^{-si_1-pj_1}&\omega_d^{-si_1-pj_2}&\omega_d^{-si_2-pj_2}&\omega_d^{-si_2-pj_2}\\
\omega_d^{-si_1+pj_1}&\omega_d^{-si_1+pj_1}&\omega_d^{-si_2+pj_2}&\omega_d^{-si_2+pj_2}\\
\end{smallmatrix} \right],
$$
which is
$$
\left [\begin{array}{cc} 
\omega_d^{si_1}&\omega_d^{si_2} \\
\omega_d^{ -si_1} & \omega_d^{-si_2}\\
\end{array} \right] \otimes \left[ \begin{array}{cc} 
\omega_d^{ pj_1} &\omega_d^{ pj_2} \\
\omega_d^{ -pj_1} & \omega_d^{-pj_2}\\
\end{array} \right]. 
$$
Since the two matrices listed above are always invertible  by given constraints on  $i_1,i_2,j_1$ and $j_2$, the Kronecker product of them is also invertible. For the cases when one of $s,p$ is zero, $\Lambda_{s,p}$ has only two pairs. It is easy to check that  the $2\times 4$ submatrix formed by $\Lambda_{s,p}$ and $\Omega$ has rank 2. In the event that $s=p=0$, it is obvious that the submatrix formed by $\Lambda_{s,p}$ and $\Omega$  has rank 1. Therefore, all conditions stated in (2) of Proposition $\ref{prop2}$ are verified. The conclusion follows. 
\item We only prove the case $\Omega=\{(i_1,0),( i_1+i_2,0),(i_1+2i_2,0),(i_1+3i_2,0): gcd(i_2,d)=1\}$, the other one follows in a similar way.  We just need to prove the case $i_1=0$ since any translation of $\Omega$ is also an admissible set by Proposition \ref{prop2}.  Given $\mathbf{a}$ has  diagonal symmetry in frequency response, we let $\Lambda_{s,p}$ be a level set of $\hat {\mathbf{a}}$ consisting of pairs $ \{(s,p),(-s,-p),(-p,-s),(p,s)\}.$ For the cases when  $s,p \neq 0$ and $s \neq p$, the submatrix built from rows indexed by $\Lambda_{s,p}$ and columns indexed by $\Omega$ is a Vandermonde matrix with 4 distinct bases:
\begin{equation}
\left( \begin{array}{cccc}
1 &\omega_d^{si_2}&\omega_d^{2si_2}&\omega_d^{3si_2} \\
1 &\omega_d^{-si_2}&\omega_d^{-2si_2}&\omega_d^{-3si_2}\\
1 &\omega_d^{pi_2}&\omega_d^{2pi_2}&\omega_d^{3pi_2}\\
1 &\omega_d^{-pi_2}&\omega_d^{-2pi_2}&\omega_d^{-3pi_2}\\
\end{array} \right).
\end{equation} So it is invertible. For the cases $s=0$ or $p=0$ or $s=p$, it is easy to verify the corresponding submatrix has full row rank.  The conclusion follows. 
\item We only prove the case $\Omega=\{i_1,i_2: gcd(\lvert i_1-i_2\vert,d)=1\} \times \{j_1, j_1+j_2,j_1+2j_2, j_1+3j_2: gcd(j_2,d)=1 \},$ the other one follows similarly.  Given $\mathbf{a}$ has  octagonal symmetry in frequency response, we let  $\Lambda_{s,p}$ be a level set of $\hat {\mathbf{a}}$ consisting of pairs \begin{center} $\{(s,p),(p,s),(-p,s),(-s,p), (-s,-p), (-p,-s),(p,-s),(s,-p)\}.$\end{center} 
For the cases when $(s,p)$ satisfies that $s\neq \pm p$ and $s,p\neq0$, we denote by $\mathbf{S}_{s,p}$ the $8 \times 8$ submatrix of $\mathbf{F}_d\otimes \mathbf{F}_d$ built from  rows indexed by $\Lambda_{s,p}$ and columns indexed by $\Omega.$ We will use a similar argument to the proof of the $\ell_0$ symmetry case. Assume $\textbf{c}=(\mathbf{c}(k,l))_{(k,l)\in \Omega}$ and $\mathbf{S}_{s,p}\textbf{c}=\mathbf{0}.$ Then for every $(s_1,p_1) \in \Lambda_{s,p}$, we have 
\begin{equation}\label{haha}
\sum\limits_{i=0}^{3}(\mathbf{c}(i_1, j_1+ij_2)\omega_d^{p_1(j_1+ij_2)})\omega_d^{s_1i_1}+ (\sum\limits_{i=0}^{3}\mathbf{c}(i_2,j_1+ij_2)\omega_d^{p_1(j_1+ij_2)})\omega_d^{s_1i_2} =0.
\end{equation} Plugging $(s_1,p_1)=(s,p),(-s,p)$ into \eqref{haha}, and note that  the matrix 
$
\left[ \begin{smallmatrix}
\omega^{ s i_1} &\omega^{ si_2} \\
\omega^{ -s i_1} &\omega^{ -si_2} \\
\end{smallmatrix}\right]
$ is invertible, we get  $\sum_{i=0}^{3}\mathbf{c}(i_{n}, j_1+ij_2)\omega^{p(j_1+ij_2)}=0$ for $n=1,2$.  Plugging
$(s_1,p_1)=(s,-p),(-s,-p)$, $(s_1,p_1)=(p,s),(-p,s)$ and $(s_1,p_1)=(p,-s),(-p,-s)$ into \eqref{haha}, we can get
\begin{equation}\label{finalequation2}
\left\{ \begin{aligned}
\sum\limits_{i=0}^{3}\mathbf{c}(i_n, j_1+ij_2)\omega_d^{-p(j_1+ij_2)}&=0,\\
\sum\limits_{i=0}^{3}\mathbf{c}(i_n, j_1+ij_2)\omega_d^{s(j_1+ij_2)}&=0,\\
\sum\limits_{i=0}^{3}\mathbf{c}(i_n, j_1+ij_2)\omega_d^{-s(j_1+ij_2)}&=0.\\
\end{aligned}
\right.
\end{equation}  for $n=1,2$.  Now we can write the above linear equations  as a linear system and solve $\mathbf{c}$. 
Since the $4 \times 4$ matrix appearing in the linear system 
$$
\left( \begin{array}{cccc}
\omega_d^{pj_1} &\omega_d^{pj_1+pj_2}&\omega_d^{pj_1+2pj_2}&\omega_d^{pj_1+3pj_2} \\
\omega_d^{-pj_1} &\omega_d^{-pj_1-pj_2}&\omega_d^{-pj_1-2pj_2}&\omega_d^{-pj_1-3pj_2}\\
\omega_d^{sj_1} &\omega_d^{sj_1+sj_2}&\omega_d^{sj_1+2sj_2}&\omega_d^{sj_1+3sj_2}\\
\omega_d^{-sj_1} &\omega_d^{-sj_1-sj_2}&\omega_d^{-sj_1-2sj_2}&\omega_d^{-sj_1-3sj_2}\\
\end{array} \right)
$$ is a nonzero constant times an invertible Vandermonde matrix , we  conclude that $\textbf{c}=\mathbf{0}$. Hence $\mathbf{S}_{s,p}$ is invertible and has full row rank. For other cases of $\Lambda_{s,p}$, it can be reduced to the case of quadrantal symmetry or diagonal symmetry and by previous results, we know the corresponding submatrix has full row rank.  Finally, It is easy to compute that $M_{\mathbf{A}}=8, N_\mathbf{A}=\frac{(d+1)(d+3)}{8}.$ By Proposition \ref{prop2}, the conclusion follows. 
\end{enumerate}
\end{proof}

Similarly, we can also provide various constructions of minimal admissible $\Omega$  consist of unions of periodic sets. We summarize them as follows.

\begin{theorem}
Let $\mathbf{A}$ be a circular convolution operator defined by a convolution kernel $\mathbf{a} \in \ell^2(\Z_d\times \Z_d)$. Suppose $d$ is odd and $d=mJ(m>1)$. Let $\Omega=\{m\Z_d+r: r\in W\subset \Z_m\}$ and $Y = \{\mathbf{e}_{i,j}, \mathbf{Ae}_{i,j},\cdots, \mathbf{A}^{m^2-1}\mathbf{e}_{i,j} : (i,j) \in \Omega \}$. 
\begin{enumerate}
\item Assume that  $\mathbf{a}$ possesses $\ell^{\infty}$ symmetry in frequency response. If $W=\{ 0,1\} \times \Z_m \cup \Z_m \times \{0,1\}$, then $Y$ is a frame for $\ell^2(\Z_d\times \Z_d)$. 
\item Assume that $\mathbf{a}$ possesses quadrantal symmetry in frequency response. Suppose the elements $i_1, i_2, j_1$ and $j_2$ of $\Z_m$ satisfy the conditions $gcd(\lvert i_1-i_2 \lvert ,m)=1$ and  $gcd(\lvert j_1-j_2 \rvert, m)=1$. If  $W=\{(i_1,j_1),(i_2,j_2),(i_1,j_2),(i_2,j_1)\}$, then  $Y$ is a frame for $\ell^2(\Z_d\times \Z_d)$. 

\item Assume that  $\mathbf{a}$ possesses  diagonal symmetry in frequency responses.  Suppose the elements $i_1, i_2, j_1$ and $j_2$ of $\Z_m$ satisfy the
conditions $gcd(i_2,m)=1$ and $gcd(j_2,m)=1$. If $W=\{(i_1,0),( i_1+i_2,0),(i_1+2i_2,0),(i_1+3i_2,0) \} \text{ or }  W=\{(0, j_1),(0,j_1+j_2),(0,j_1+2j_2), (0,j_1+3j_2) \}$, then $Y$ is a frame for $\ell^2(\Z_d\times \Z_d)$. 

\item Assume that $\mathbf{a}$ possesses  octagonal symmetry in frequency response. Suppose the elements $i_1, i_2, j_1$ and $j_2 $ of $\Z_m$ satisfy the
conditions $gcd(\lvert i_1-i_2\vert,m)=1$ and $gcd(j_2,m)=1$. If $W=\{i_1,i_2 \} \times \{j_1, j_1+j_2,j_1+2j_2, j_1+3j_2 \}$, then $Y$ is a frame for $\ell^2(\Z_d\times \Z_d)$. Alternatively, suppose the elements $i_1, i_2, j_1$ and $j_2 $ of $\Z_m$ satisfy the
conditions $gcd(\lvert j_1-j_2\vert,m)=1$ and $gcd(i_2,m)=1$. If $W=\{i_1, i_1+i_2,i_1+2i_2, i_1+3i_2\} \times \{j_1,j_2\}$, then $Y$ is a frame for $\ell^2(\Z_d\times \Z_d)$. 

\end{enumerate}
\end{theorem}
\begin{proof}
From the proof of Theorem \ref{main2}, we see that it suffices to deal with certain submatrices of $\mathbf{F}_m\otimes \mathbf{F}_m$. We can then follow a similar argument as the proof of Theorem \ref{stablerecovery4case}.

\end{proof}

\section{General case}\label{finiteabeliangroup}
In this section, we will briefly discuss Problem \ref{prob1} for a  finite abelian group $I$ and  generalize  the algebraic characterizations given in Proposition \ref{prop1} and Proposition \ref{prop2}.  We will illuminate the relations between the problem under consideration,  representations of $I$ that commute with the evolution operator $\mathbf{A}$, and the character table of $I$. We will recall several classical results of representation theory of finite abelian groups.

 We assume $I=\{i_0,i_1,\cdots, i_{d-1}\}$, where $d$ is a positive integer in this section.  A character of a finite abelian group $I$ is a group homomorphism $\chi: I\rightarrow S^1$.  The characters of  $I$ form a finite abelian group with respect to the pointwise product.  This group  is called the character group of $I,$ and denoted by $\hat {I}$. In fact, we can prove that $I$ is isomorphic to $\hat I$. Thus, $I$ can serve as index set for $\hat I$ and we assume $\hat {I}=\{\chi_{i_0},\chi_{i_1},\cdots, \chi_{i_{d-1}}\}$. We choose an enumerate of $\hat I$ such that $\chi_{i_0}$ is the trivial character of $I$. Let $\mathbf{X}=\big(\chi_{i_s}(i_t)\big)$ be the matrix whose rows are indexed by $\hat I$ and columns are indexed by $I$.  Then $\mathbf{X}$ is a square matrix of dimension $d$ and it is called the character table of group $I$.  We can further show that the matrix $\frac{1}{\sqrt{d}}\mathbf{X}$ is unitary.

Let us consider the signal space $\ell^2(I)$. The distinguishing feature of functions on  $I$ is that the group acts on itself by left translation, thereby moving around the functions on it.  More specifically, for $i \in I$, the translate $\mathbf{T}_i \mathbf{f}$ of a function $\mathbf{f}$ by $i$ is the function on $I$ defined by  $$(\mathbf{T}_i \mathbf{f})(j)=f(i+j), \text{ for function }  \mathbf{f} \in \ell^2(I), \text{ and } i, j \in I.$$ It is easy to check that the map $\rho: i \rightarrow \mathbf{T}_i$ from $I$ to the $\mathbb{C}-$linear automorphisms of $\ell^2(I)$ is a group homomorphism.  The pair $(\rho, \ell^2(I))$ is called the \textit{left regular representation} of $I$. As we can see, a character $\chi_{i_s}$ of $I$ is indeed a square summable function defined on $I$. Using the fact that $\chi_{s}$ is a group homomorphism, it follows that  $\chi_{s}$ is an eigenvector   for  all translation operators $\{\mathbf{T}_i :i\in I\}$.  Note that we have $d$ group characters, so  the translation operators $\{\mathbf{T}_i :i\in I\}$ are \textit {simultaneously diagonalizable} by characters of $I$. 

By an abuse of notation, suppose $\mathbf{A}$ is an evolution operator on $\ell^2(I)$ such that $\mathbf{f}$ is evolving under the iterated action of $\mathbf{A}$. This discrete evolution process is called \textit{spatially invariant} if  the operator $\mathbf{A}$ commutes with any spatial translations on $\ell^2(I)$. It turns out that $\mathbf{A}$ is indeed a circular convolution operator. The following theorem can be found in \cite[Theorem 5.1.3] {BL2009}:

\begin{theorem} \label{commute} Let $\mathbf{A}$ be a linear operator on $\ell^2(I)$. Then $\mathbf{A}$ commutes with spatial translations $\{\mathbf{T}_i: i\in I\}$ if and only if $\mathbf{A}$ is a circular convolution operator with some convolution kernel $\mathbf{a} \in \ell^2(I)$. 
\end{theorem}

Now let $\mathbf{a} \in \ell^2(I)$, we definite $\mathbf{\hat a} \in \ell^2(\hat I)$ by 
$$\mathbf{\hat a}(\chi_{s})=\langle \mathbf{a}, \chi_s\rangle, s \in I.$$  Suppose $\mathbf{A}$ is a circular convolution operator with kernel $\mathbf{a}$. It follows from Theorem $\ref{commute}$ that  the characters of $I$ are also eigenvectors of $\mathbf{A}$. By computation,  we can obtain
\begin{equation}\label{generall}
\mathbf{A}= (\frac{1}{\sqrt{d}}\mathbf{X})^T \mathbf{diag(\hat a)} \frac{1}{\sqrt{d}}  \overline{\mathbf{X}},
\end{equation}
where $\overline{\mathbf{X}}$ denotes the complex conjugate of the matrix $\mathbf{X}$. In the case of  $I=\Z_d$ and $\hat I=\{\chi_0,\chi_1,\cdots, \chi_{d-1}\}$, where $\chi_s(1)=e^{\frac{2\pi i s}{d}}$ for $s \in \Z_d$, we can show that $\frac{1}{\sqrt{d}}  \overline{\mathbf{X}}=\mathbf{F}_d$  and  the equation $\eqref{generall}$ is exactly the same with the equation \eqref{equation1}.  Now we are ready to extend Proposition \ref{prop1} and Proposition \ref{prop2} to the general case. 
\begin{proposition}
Let $\mathbf{A}$ be a circular convolution operator with a kernel $\mathbf{a} \in \ell^2(I)$ and $\{\Lambda_k: k=1,\cdots, N_{\mathbf{A}}\}$ be the level sets of ${\mathbf{\hat a}}$. Suppose that  $\Omega \subset I$  and  Let 
 $Y = \{\mathbf{e}_{i}, \mathbf{Ae}_{i},\cdots, \mathbf{A}^{l_i}\mathbf{e}_{i} : i \in \Omega \}$. 
 \begin{enumerate}
\item If $Y$ is a frame for $\ell^2(I)$, then for each $k$,  the submatrix of $\overline{\mathbf{X}}$ built from rows indexed by $\Lambda_k$ and columns indexed by $\Omega$  has rank $|\Lambda_k|$.  
\item For each $i \in \Omega $, we define  by $r_i$  the degree of the $ \mathbf{A}$-annihilator of $\mathbf{e}_{i}$. If for each $i \in \Omega$ and each $k$, $l_{i} \geq r_{i}-1$ and   the submatrix of $\overline{\mathbf{X}}$ built from rows indexed by $\Lambda_k$ and columns indexed by $\Omega$  has rank $|\Lambda_k|$, then  $Y$ is a frame for $\ell^2(I)$.

\item If for each $i \in \Omega$ and each $k$, $\mathbf{A}^{l_i+1} \mathbf{e}_{i}$ is in the space spanned by $Y$ and the submatrix of $\overline{\mathbf{X}}$ built from rows indexed by $\Lambda_k$ and columns indexed by $\Omega$ has rank $|\Lambda_k|$,  then  $Y$ is a frame for $\ell^2(I)$. 
\end{enumerate}
\end{proposition}

An immediate corollary we can get is
\begin{corollary} Let $\mathbf{A}$ be a circular convolution operator with a kernel $\mathbf{a} \in \ell^2(I)$.  Suppose  $\Omega \subset I$. If the set $\Omega$ 
is an admissible set for $\mathbf{A}$, then so is any translation of $\Omega$. 

\end{corollary}

\section{Concluding Remarks}\label{conclusion}

In this paper, we have characterized universal spatiotemporal sampling sets for discrete spatially invariant evolution systems. In the one variable case $I=\Z_d$, we have shown that the lower bound of sensors numbers we derived is achievable and how to construct minimal universal irregular and unions of periodic sensor locations  for  convolution operators with eigenvalues  subject to the same largest geometric multiplicity. In the two variable case $I=\Z_d\times \Z_d$, we have shown that the lower bound of sensor numbers  derived may not be achievable and the problem of finding universal spatiotemporal sampling sets is less favorable. We restricted ourselves to the evolution systems in which the convolution operators have certain symmetries in the Fourier domain and gave various constructions of minimal universal irregular and unions of periodic admissible sets. The ideas can be easily generalized to the three or higher variable case $I=\Z_{d_1}\times \cdots \times \Z_{d_n}$. As we have already seen, the problem of finding a deterministic admissible set $\Omega$ with $\lvert \Omega\rvert$ as few as possible for a 2D convolution operator $\mathbf{A}$ is not easy in general; however, if we relax the condition of finding full spark matrices to finding matrices with \textit {restricted isometry property} (the definition can be found in \cite{PGH11} ) from discrete Fourier matrices, our algebraic perspective allows us to obtain random constructions of universal (but not minimal) admissible sets  with high probability, which are parallel to those classical results in the literature of compressed sensing, e.g., \cite{CRT06}. The results in the unions of periodic case allow us to obtain near-optimal deterministic constructions of  admissible sets $\Omega$ by choosing the period $m$ to be prime and making use of Weil sum; see related results in \cite{X11}.  Finally, we have also studied the general finite abelian group case and established a  connection between the problem under consideration with the character table of $I$.

\section{Acknowledgements}
The author would like to sincerely thank Akram Aldroubi for his careful reading of the manuscript and insightful comments. The author is indebted to anonymous reviewers for their valuable advice to improve the paper substantially. The author  would also like to thank  Xuemei Chen, Keaton Hamm and James Murphy for their useful suggestions to revise the presentation of the paper. 
\bibliography{DynamicalsamplingRevised2}{}
\bibliographystyle{plain}

\end{document}